\documentclass[12pt]{article}

\usepackage[a4paper,margin=30mm]{geometry}

\usepackage{amsmath,amssymb,amsfonts,amsthm}
\usepackage{mathrsfs}
\usepackage{graphicx}
\usepackage[percent]{overpic}
\usepackage{booktabs}
\usepackage{multirow}
\usepackage{xcolor}
\usepackage{mathtools}
\usepackage{cite}
\usepackage{hyperref}
\usepackage{comment}
\usepackage{stmaryrd}

\usepackage{tikz}
\usepackage{pgfplots}
\usetikzlibrary{math,calc,automata,arrows,positioning}
\pgfplotsset{compat=1.18}

\usepackage{authblk}

\newtheorem{thm}{Theorem}[section]
\newtheorem{prop}[thm]{Proposition}
\newtheorem{cor}[thm]{Corollary}
\newtheorem{lem}[thm]{Lemma}

\allowdisplaybreaks[4]
\newcommand{\la}{\langle}
\newcommand{\ra}{\rangle}
\newcommand{\ldbr}{\llbracket}
\newcommand{\rdbr}{\rrbracket}

\title{Stationary Measures and Mean Flux \\
Depending on Multiple Conserved Quantities \\
in a Stochastic Cellular Automaton}

\author[1]{Kazushige Endo\thanks{Corresponding author: \texttt{k-endo@math.kindai.ac.jp}}}
\author[2]{Shinsuke Iwao}
\author[3]{Hidetomo Nagai}
\author[4]{Yushi Nakano}

\affil[1]{Department of Mathematics, Kindai University, Osaka, Japan}
\affil[2]{Faculty of Business and Commerce, Keio University, Kanagawa, Japan}
\affil[3]{Department of Mathematics, Tokai University, Kanagawa, Japan}
\affil[4]{Faculty of Science, Hokkaido University, Hokkaido, Japan}

\date{}

\begin{document}

\maketitle

\begin{abstract}    
We analyze a stochastic 5-neighbor cellular automaton with several conserved quantities, including the particle density.
By examining the eigenvalue problem of the associated transition matrix, we derive an explicit formula for the stationary distribution on each irreducible component, in which the weight of each configuration is expressed in terms of the numbers of occurrences of two specific local patterns. This analysis further allows us to theoretically derive the dependence of the mean flux on the conserved quantities. In particular, we recover the mean flux formula in the deterministic case by taking the zero-noise limit of the system.
\end{abstract}

\section{Introduction}\label{1}
Cellular automata have been studied for various fields such as physics, engineering, and mathematics. 
In particular, the dependence of the particle momentum on particle density in their asymptotic behavior has been one of the central topics in the analysis of cellular automata \cite{fuks, nishinari}.
For example, the asymmetric simple exclusion process (ASEP) is a popular stochastic cellular automaton which is a multibody random walk model where each particle moves stochastically along one-dimensional lattice space\cite{spitzer, derrida1, derrida2}. 
The stationary measure of ASEP has been investigated in connection with orthogonal polynomials \cite{sasamoto}.
Another important example is the Nagel-Schreckenberg model, which is related to traffic flow \cite{nagel, ns}.
Nagel and Schreckenberg theoretically analyzed the mechanism of traffic jams, which depends on the density of cars in the traffic system.

In the present paper, we investigate the asymptotic behavior of a stochastic 5-neighbor cellular automaton and analyze the dependence of its mean flux on \emph{two} conserved quantities.
In contrast to the ASEP, where the mean flux is uniquely determined by the particle density, the mean flux in our system depends not only on the particle density but also on an additional conserved quantity.

First, we recall the properties of a deterministic cellular automaton with 5-neighbors investigated in \cite{3d}, described by
\begin{equation}
\label{dsystem}
u_j^{n+1}=u_j^n+q_0\left(u_{j-2}^n,u_{j-1}^n,u_j^n,u_{j+1}^n\right)-q_0\left(u_{j-1}^n,u_j^n,u_{j+1}^n,u_{j+2}^n\right).
\end{equation}
Here, the variable $u$ is binary, taking values $0$ or $1$, where $j$ denotes the site index and $n$ denotes the discrete time step.
The function $q_0$ represents the flux, which is specified in Table~\ref{flux}.
\begin{table}[h]
\caption{Rule table of $q_0(w,x,y,z)$ in (\ref{dsystem}). Upper and lower rows denote $(w,x,y,z)$ and $q_0(w,x,y,z)$ respectively.}
\label{flux}
\begin{center}
\begin{tabular}{r}
\begin{tabular}{|c||c|c|c|c|c|c|c|c|}
\hline
($w,x,y,z$)
& 1111 & 1110 & 1101 & 1100 & 1011 & 1010 & 1001 & 1000 \\
\hline
$q_0(w,x,y,z)$
& 1 & 1 & 1 & 1 & 0 & 0 & 0 & 0 \\
\hline
\end{tabular}
\medskip\\
\begin{tabular}{|c|c|c|c|c|c|c|c|}
\hline
0111 & 0110 & 0101 & 0100 & 0011 & 0010 & 0001 & 0000 \\
\hline
0 & 1 & 0 & 0 & 0 & 0 & 0 & 0 \\
\hline
\end{tabular}
\end{tabular}
\end{center}
\end{table}
We assume the periodic boundary condition for space sites with a period $L$. 
From the evolution equation above, it is easily shown that 
\begin{equation}
\label{law}
\sum_{j=1}^{L}u_j^{n+1}=\sum_{j=1}^{L}u_j^n, \ \ \ \sum_{j=1}^{L}{u_{j-1}^{n+1}u_{j}^{n+1}\left(1-u_{j+1}^{n+1}\right)}=\sum_{j=1}^{L}{u_{j-1}^{n}u_{j}^{n}\left(1-u_{j+1}^{n}\right)}.
\end{equation}
We use the notation 
\[
\# a_1 a_2\cdots a_k(x)
\]
to denote the number of occurrences of the local pattern $a_1 a_2\cdots a_k$ in a configuration $x=(x_j)_{j=1}^L\in \{0,1\}^L$ among the $L$ sites.
The density of such patterns in $x$ is denoted by $\rho_{a_1 a_2\cdots a_k}(x)$, i.e., 
\[
\rho_{a_1 a_2\cdots a_k}(x) = \frac{\# a_1 a_2\cdots a_k(x)}{ L}.
\]
Then, the relations in (\ref{law}) show that $\rho_1$ and $\rho_{110}$ are conserved quantities:
\[
\rho_1(u^n)=\rho_1(u^0), \quad \rho_{110}(u^n)=\rho_{110}(u^0)
\]
for every $n\ge 0$, where $u^n=(u^n_j)_{j=1}^L$.
We simply write $\rho_1$, $\rho_{110}$ for $\rho_1(u^0)$, $\rho_{110}(u^0)$, respectively, if it makes no confusion.

Table~\ref{flux} corresponds to the following particle motion rules:
\begin{itemize}
\item
An isolated particle (010) remains unchanged.
\item
For a pair of adjacent two particles (0110), both particles move.
\item
For a local pattern of more than two particles ($011\ldots10$), all particles except the leftmost one move.
\end{itemize}
Figure~\ref{te1} shows an example of the dynamics of this system. 
\begin{figure}[h]
\centering
  \begin{overpic}[width=70mm,trim= -20pt -20pt -20pt -20pt, clip=false]{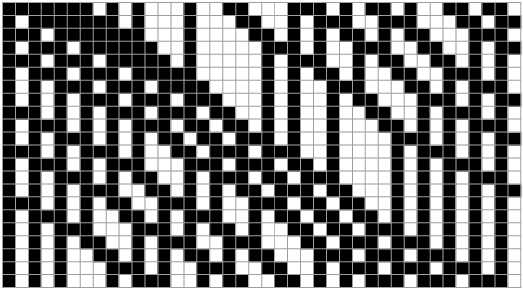} 
    \put(2, 59){\vector(1, 0){40}} 
    \put(47, 58){\small $j$}
    \put(2, 59){\vector(0, -1){40}} 
    \put(1, 12){\small $n$}
  \end{overpic}
\caption{Example of time evolution of (\ref{dsystem}). Black squares $\blacksquare$ mean $u=1$ and white squares $\square$ $u=0$.}
\label{te1}
\end{figure}
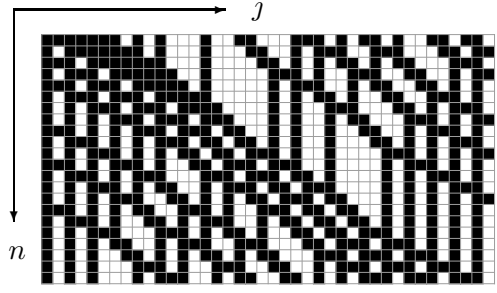

The mean flux across spatial sites in the asymptotic regime, representing the average particle momentum in the long-time limit, is defined by
\begin{equation}
Q_0=\lim_{n\to \infty}{\frac{1}{L}\sum_{j=1}^{L}q_0\left(u_{j-2}^n,u_{j-1}^n,u_j^n,u_{j+1}^n\right)}.
\end{equation}
The limit is known to exist 
and to depend uniquely on the pair of the conserved densities $\rho_1$ and $\rho_{110}$ as follows \cite{3d}:
\begin{equation}
Q_0=\max{\left(2\rho_1-1, 2\rho_{110} \right)}.
\label{qd}
\end{equation}
Figure~\ref{fdd} depicts the three-dimensional `fundamental diagram' obtained by (\ref{qd}).
The domain $(\rho_1, \rho_{110})$ is constrained by $2\rho_{110} \le \rho_1 \le 1-\rho_{110}$, reflecting the relationship between $\# 1$ and $\# 110$.
Typically, the fundamental diagram is described by the relationship between mean flux and density.
However, since the mean flux of this system depends on two independent quantities $\rho_1$ and $\rho_{110}$, the diagram naturally extends to three dimensions.
\begin{figure}[h]
  \centering
  \begin{overpic}[width=80mm,
   trim= -30pt -30pt -30pt -30pt, clip=false]{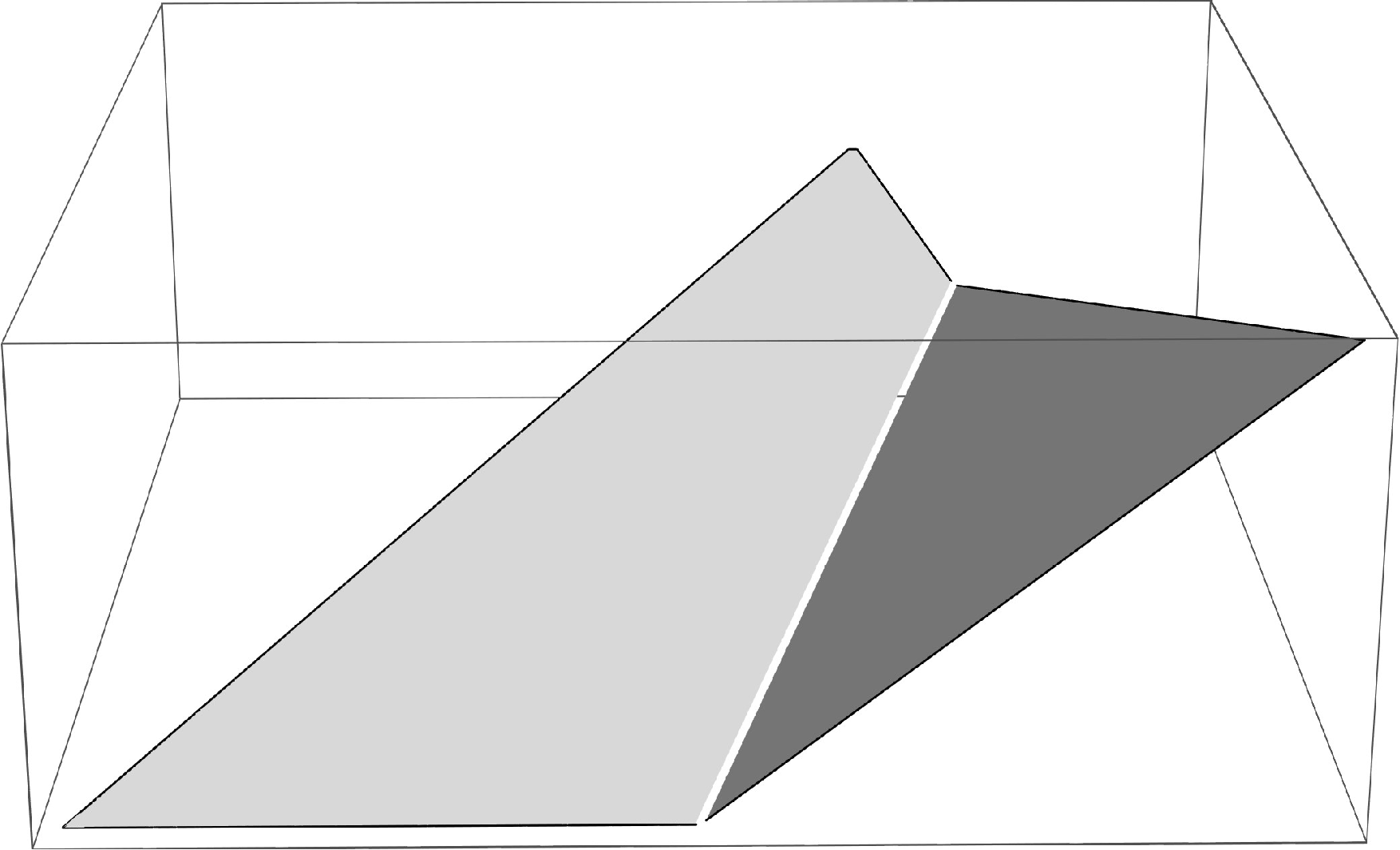}
    \put(47,54){\scriptsize $(2/3,1/3,2/3)$}
    \put(66,43){\scriptsize $(3/4,1/4,1/2)$}
    \put(0,36){\footnotesize $1$}
    \put(0.5,0){\footnotesize $0$}
    \put(95,0){\footnotesize $1$}
    \put(9,60){\scriptsize $\frac{1}{3}$}
    \put(0,48){\footnotesize $\rho_{110}$}
    \put(50,0){\footnotesize $\rho_1$}
    \put(-2,18){\footnotesize $Q_0$}
  \end{overpic}
\caption{Fundamental diagram of (\ref{dsystem}).}
\label{fdd}
\end{figure}
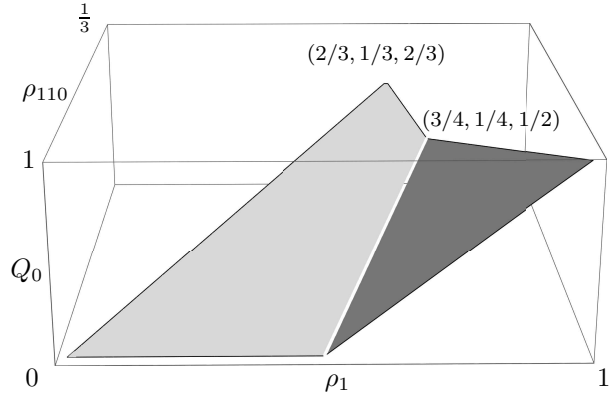

In the present paper, we introduce a stochastic parameter to the above deterministic system and analyze the asymptotic distribution of the stochastic system.
Based on the analysis, the expected value of the mean flux shall be rigorously derived.
Moreover, we confirm that the theoretical formula (\ref{qd}) in the deterministic case can be obtained as the zero-noise limit.

{\bf Structure of the paper.} In Section \ref{2}, we propose a stochastic extension of the cellular automaton (\ref{flux}) and express its mean flux in terms of two local-pattern densities.
After preparation in Section \ref{3} through examples, in Section \ref{s:m}, we give the stationary distribution of the stochastic system, which is derived by eigenvalue problems of transition matrices \cite{endo}. Furthermore, we take the noise-zero limit to obtain the deterministic profile of the fundamental diagram shown in Figure~\ref{fdd}.
Finally, in Section \ref{s:pf}, we provide the proof of the main theorem.

\section{Stochastic 5-neighbor cellular automaton} \label{2}
We introduce an external random variable $a$ into Table~\ref{flux} while preserving the conservation laws of (\ref{law}) for $\rho_1$ and $\rho_{110}$, thereby obtaining a stochastic cellular automaton given by 
\begin{equation}
\label{eqs}
u_j^{n+1}=u_j^n+q_1\left(u_{j-2}^n,u_{j-1}^n,u_j^n,u_{j+1}^n\right)-q_1\left(u_{j-1}^n,u_j^n,u_{j+1}^n,u_{j+2}^n\right)
\end{equation}
and Table~\ref{fluxs}. 
\begin{table}[h]
\caption{Rule table of $q_1(w,x,y,z)$ of (\ref{eqs}).}
\label{fluxs}
\begin{center}
\begin{tabular}{r}
\begin{tabular}{|c||c|c|c|c|c|c|c|c|}
\hline
($w,x,y,z$)
& 1111 & 1110 & 1101 & 1100 & 1011 & 1010 & 1001 & 1000 \\
\hline
$q_1(w,x,y,z)$
& 1 & 1 & 1 & 1 & 0 & 0 & 0 & 0 \\
\hline
\end{tabular}
\medskip\\
\begin{tabular}{|c|c|c|c|c|c|c|c|}
\hline
0111 & 0110 & 0101 & 0100 & 0011 & 0010 & 0001 & 0000 \\
\hline
$a$ & 1 & 0 & 0 & 0 & 0 & 0 & 0 \\
\hline
\end{tabular}
\end{tabular}
\end{center}
\end{table}

\noindent
The variable $a$ is a random variable defined by
\begin{equation}
a=
\begin{cases}
0 & \text{(with probability $\alpha$)} \\
1 & \text{(with probability $1-\alpha$).}
\end{cases}
\end{equation}

Figure~\ref{ste} shows an example of time evolution, and Figure~\ref{sim} shows numerical results of the mean flux $Q_u$ of the stochastic system.
Compared to Figure~\ref{fdd}, the mean flux surface of the stochastic system converges to that of the deterministic system (\ref{flux}) as $\alpha \to 1$. Figure~\ref{ss} shows cross-sectional plots at $\rho_{110}=7/60$ and $\alpha=0.5, \ 0.7, \ 0.9, \ 1$. 

\begin{figure}[h]
\centering
  \begin{overpic}[width=70mm,trim= -20pt -20pt -20pt -20pt, clip=false]{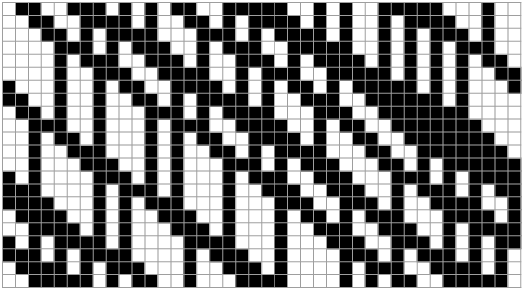} 
    \put(2, 59){\vector(1, 0){40}} 
    \put(47, 58){\small $j$}
    \put(2, 59){\vector(0, -1){40}} 
    \put(1, 12){\small $n$}
  \end{overpic}
\caption{Example of time evolution of (\ref{eqs}) for $\alpha=0.5$.}
\label{ste}
\end{figure}
\begin{figure}[h]
  \centering
  \begin{overpic}[width=80mm,
   trim= -10pt -10pt -10pt -10pt, clip=false]{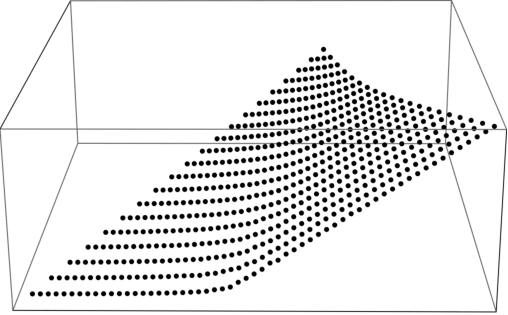}
    \put(0,36){\footnotesize $1$}
    \put(1,0.5){\footnotesize $0$}
    \put(95,0){\footnotesize $1$}
    \put(9,60){\scriptsize $\frac{1}{3}$}
    \put(0,48){\footnotesize $\rho_{110}$}
    \put(50,0){\footnotesize $\rho_1$}
    \put(-2,18){\footnotesize $Q_u$}
  \end{overpic}
\caption{Numerical results of fundamental diagram of (\ref{eqs}) for $L=60, \ \alpha=0.7$ averaged from $n=0$ to 3000.}
\label{sim}
\end{figure}
\begin{figure}[h]
\centering
  \begin{overpic}[width=90mm,trim= -15pt -15pt -15pt -15pt, clip=false]{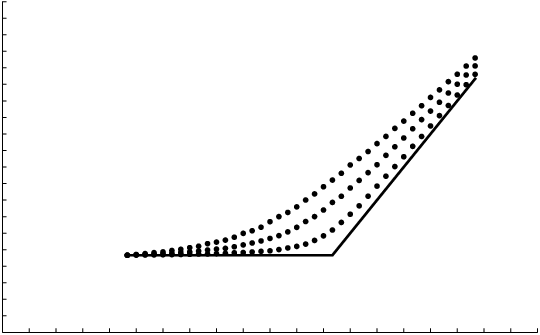} 
    \put(3,64){$Q_u$}

    \put(97,5){$\rho_1$}

    \put(3,0){\scriptsize 0.0}
    \put(21,0){\scriptsize 0.2}
    \put(39,0){\scriptsize 0.4}
    \put(57,0){\scriptsize 0.6}
    \put(75,0){\scriptsize 0.8}
    \put(93,0){\scriptsize 1.0}

    \put(-1, 4){\scriptsize 0.0}
    \put(-1,15){\scriptsize 0.2}
    \put(-1,26){\scriptsize 0.4}
    \put(-1,37){\scriptsize 0.6}
    \put(-1,48){\scriptsize 0.8}
    \put(-1,59){\scriptsize 1.0}
  \end{overpic}
\caption{Example of fundamental diagram of (\ref{eqs}) for $L=60, \ \rho_{110}=7/60, \\
\ 7/30 \le \rho_1 \le 53/60$. Dots are numerical results for $\alpha=0.5, \ 0.7, \ 0.9$ averaged from $n=0$ to 3000, and the curve is for $\alpha=1$, which is obtained by (\ref{qd}).}
\label{ss}
\end{figure}

In order to analyze this stochastic system, we express the time evolution using a new variable $v_j^n$ given by 
\begin{equation}\label{eqs2}
v_j^n=u_{n-j}^n.
\end{equation}
Then, the time evolution equation takes the form\footnote{
Note that, by \eqref{eqs} and \eqref{eqs2},
\begin{align*}
v_j^{n+1} = u_{n+1-j}^{n+1} 
&= u_{n+1-j}^n+q_1\left(u_{n-1-j}^n,u_{n-j}^n,u_{n+1-j}^n,u_{n+2-j}^n\right)-q_1\left(u_{n-j}^n,u_{n+1-j}^n,u_{n+2-j}^n,u_{n+3-j}^n\right)\\
&=v_{j-1}^n+v_{j}^n-v_{j}^n+q_1\left(v_{j+1}^n,v_{j}^n,v_{j-1}^n,v_{j-2}^n\right)-q_1\left(v_{j}^n,v_{j-1}^n,v_{j-2}^n,v_{j-3}^n\right).
\end{align*}
On the other hand, it is straightforward to check that 
\[
q_1(z,y,x,w)-y=-q(w,x,y,z)
\]
for every $(w,x,y,z)$.
The desired equality \eqref{eqsg} immediately follows from these observations.
}
\begin{equation}
\label{eqsg}
v_j^{n+1}=v_j^n+q\left(v_{j-3}^n,v_{j-2}^n,v_{j-1}^n,v_{j}^n\right)-q\left(v_{j-2}^n,v_{j-1}^n,v_{j}^n,v_{j+1}^n\right),
\end{equation}
and $q(w,x,y,z)$ is given by Table~\ref{fluxsg}.
\begin{table}[h]
\caption{Rule table of $q(w,x,y,z)$ of (\ref{eqsg})}
\label{fluxsg}
\begin{center}
\begin{tabular}{r}
\begin{tabular}{|c||c|c|c|c|c|c|c|c|}
\hline
($w,x,y,z$)
& 1111 & 1110 & 1101 & 1100 & 1011 & 1010 & 1001 & 1000 \\
\hline
$q(w,x,y,z)$
& 0 & $b$ & 0 & 0 & 0 & 1 & 0 & 0 \\
\hline
\end{tabular}
\medskip\\
\begin{tabular}{|c|c|c|c|c|c|c|c|}
\hline
0111 & 0110 & 0101 & 0100 & 0011 & 0010 & 0001 & 0000 \\
\hline
0 & 0 & 0 & 0 & 0 & 1 & 0 & 0 \\
\hline
\end{tabular}
\end{tabular}
\end{center}
\end{table}

\noindent The random variable $b$ is defined by
\begin{equation}
b=
\begin{cases}
1 & \text{(with probability $\alpha$)} \\
0 & \text{(with probability $1-\alpha$).}
\end{cases}
\end{equation}

The motion rule of this stochastic system is as follows.
\begin{itemize}
\item
An isolated particle (010) moves to its neighboring right empty site.
\item
For a pair of adjacent two particles (0110), both remain at their positions.
\item
For a pattern of more than two particles ($011\ldots10$), the rightmost particle moves to its neighboring right empty site with probability $\alpha$, while the other particles remain at their positions.
\end{itemize}
\begin{figure}[h]
\centering
  \begin{overpic}[width=70mm,trim= -20pt -20pt -20pt -20pt, clip=false]{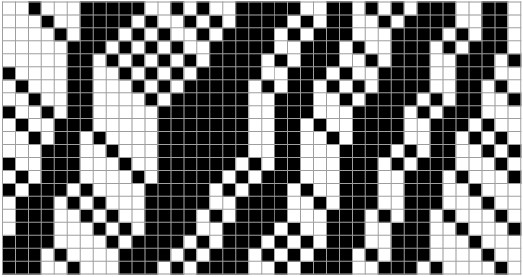} 
    \put(2, 59){\vector(1, 0){40}} 
    \put(47, 58){\small $j$}
    \put(2, 59){\vector(0, -1){40}} 
    \put(1, 12){\small $n$}
  \end{overpic}
\caption{Example of time evolution of (\ref{eqsg}) for $\alpha=0.5$.}
\label{steg}
\end{figure}
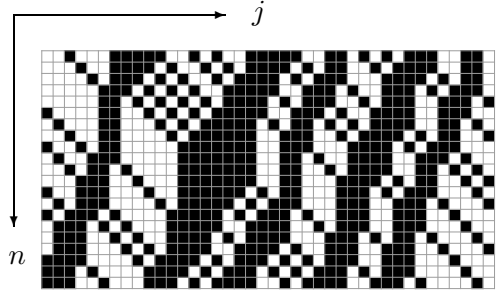
Figure~\ref{steg} shows an example of the time evolution of the transformed stochastic system.

The mean flux $Q$ of the stochastic system \eqref{eqsg} is given by
\begin{equation}
Q=\lim_{n\to \infty}{\frac{1}{L}\sum_{j=1}^{L}E[q\left(v_{j-2}^n,v_{j-1}^n,v_j^n,v_{j+1}^n\right)}],
\label{Q2}
\end{equation}
where $E[\cdot ]$ is the expectation according to the choice of $b$ in each step
(which may depend on the initial configuration $v^0$), if the limit exists.
Notice that, based on the particle motion rule above, the mean flux can be expressed in terms of the local densities $\rho_{1110}$ and $\rho_{010}$ as
\begin{equation}
Q=\lim_{n\to \infty}{E\left[\alpha\rho_{1110}(v^n)+\rho_{010}(v^n)\right]}.
\label{Q}
\end{equation}
Notice also that, using the inverse transformation $u_{n-j}^n=v_j^n$, it holds that \[
Q_u=\rho_1-Q
\]
for initial configurations $u^0$, $v^0$ satisfying $u_{-j}^0=v_j^0$.

\section{Definitions and Examples} \label{3}
Let us consider the stochastic process $(v^n)_{n\ge 0}$ governed by the update rule \eqref{eqsg} with Table~\ref{fluxsg}.
This process defines a Markov chain on the finite configuration space $\{0,1\}^L$, where a configuration here refers to a vector of length $L$ with binary entries. 
It is well known that any finite-state Markov chain can be decomposed into disjoint classes, called irreducible classes. 
To be more precise in our setting, let $f_n(x',x)$ denote the $n$-th transition probability from a configuration $x'$ to another configuration $x$ determined by Table~\ref{fluxsg} (so that the probability that $v^{m+n}=x$ when $v^m=x'$ is equal to $f_n(x',x)$ for any integer $n,m\ge 0$ and $x',x\in\{0,1\}^L$).
We simply write $f(x',x)$ for $f_1(x',x)$.
We say that $\Omega \subset \{0,1\}^L$ is  \emph{irreducible} if for any $x',x\in \Omega$, there is an $n\ge 0$ such that $f_n(x',x)>0$.
If an irreducible class $\Omega$ is closed, that is, if there exists no
pair $(x',x)$ with $x'\in\Omega$ and
$x\in\{0,1\}^L\setminus\Omega$ such that $f(x',x)>0$, then $\Omega$ is
called \emph{recurrent}; otherwise, it is called \emph{transient}.
Furthermore, a recurrent irreducible class $\Omega$ is called \emph{aperiodic} if the period $r(\Omega):=\gcd \{n:f_n(x,x)>0,\; x\in\Omega \}$ of $\Omega$ is $1$.
A probability vector $p(x)$ on $\{0,1\}^L$ is said to be \emph{stationary} if 
\begin{equation}  \label{eq:1110goal}
 \mathcal Ap(x)=p(x)  \quad \text{for any $x\in \{0,1\}^L$},
\end{equation}
where $\mathcal A\nu$ is a probability vector given by
\begin{equation} 
  \mathcal A\nu (x)=\sum _{x'\in \{0,1\}^L} f(x',x)\nu (x')
\end{equation}
for each probability vector $\nu$
(notice that if $\nu=P(v^m=\cdot)$, then $\mathcal A^n\nu = P(v^{n+m}=\cdot )$).
The following holds from the general theory of finite-state Markov chains (cf. Propositions 1.28, 1.29, Exercise 1.6 and Theorem 4.9 in~\cite{LP2017}).
\begin{prop}\label{prop:0617}
The configuration space $\{0,1\}^L$ can be decomposed into disjoint finitely many recurrent irreducible components $\Omega _1,\Omega _2,\ldots ,\Omega_k$ and finitely many transient irreducible components, such that
each recurrent irreducible component $\Omega_j$ admits a unique stationary distribution $p_j(x)$.
Furthermore, for each $j\in\{1,2,\ldots,k\}$, there exist probability vectors $h_{j,1},h_{j,2},\ldots ,h_{j,r(j)}$ with mutually disjoint support on $\Omega_j$ with $r(j)=r(\Omega_j)$, positive bounded linear functionals $\eta_{j,1},\eta_{j,2},\ldots ,\eta_{j,r(j)}$ on the space of all probability vectors of $\{0,1\}^L$ and constants $C>0$, $\kappa \in (0,1)$ satisfying that
\begin{enumerate} 
\item $\mathcal Ah_{j,i}=h_{j,i +1 \; \mathrm{mod} \; r(j)}$ for each $i\in\{1,2,\ldots,r(j)\}$, $j\in\{1,2,\ldots,k\}$;
\item $\frac{1}{r(j)}\sum_{i=1}^{r(j)}h_{j,i} =p_j(x)$  for each $j\in\{1,2,\ldots,k\}$;
\item For any probability vector $\nu$ on $\Omega$ and $n\ge 0$,
\[
\max_{x\in \Omega}\Big|\mathcal A^n \nu (x) - \sum_{j=1}^k\sum_{i=1}^{r(j)} \eta _{j,i}(\nu) h_{j,i+n\; \mathrm{mod}\; r(j)}(x)\Big\vert \le C\kappa ^n .
\]
\end{enumerate}
\end{prop}

Notice that in the deterministic case $\alpha=0$, the transient components of Proposition \ref{prop:0617} may not be empty (for example, when $x=(111100)$, $f(x',x)=0$ for any $x'\in\{0,1\}^L$ and any irreducible set including $x$ is transient).
In contrast to it, one can decompose $\{ 0,1\} ^L$ into recurrent irreducible components for the stochastic 5-neighbor cellular automaton of \eqref{eqsg} as follows.
\begin{lem}
\label{lem:0630}
For any $x\in\{0,1\}^L$, the set $B(x):=\{x'\in \{0,1\}^L : f(x',x)>0\}$ is nonempty.
Moreover, 
    $B(x)=\{x\}$
if and only if either $\# 1(x)\in\{0,L\}$ or $\# 010(x)=\# 111(x) =0$.
In particular, every irreducible component $\Omega$ of the Markov chain defined by \eqref{eqsg} with Table~\ref{fluxsg} is recurrent.
\end{lem}
\begin{proof}
The ``if'' part
is straightforward since the local pattern `0110' persists with probability $1$.
Thus, we shall prove the converse.

Assume that $\# 1(x)\in[1,L-1]$, and $\# 010(x)\ge 1$ or $\# 111(x) \ge 1$.
In the case $\# 010(x)\ge 1$, take $x'$ such that $x_j'x_{j+1}'=10$ whenever $x_jx_{j+1}x_{j+2}=010$ (the subindices are modulo $L$), and $x_i=x_i'$ for other $i$.
Then, $f(x',x)>0$ because 
\begin{itemize}
    \item if $x_{j-1}x_jx_{j+1}x_{j+2}=0010$, then $x_{j-1}'x_j'x_{j+1}'x_{j+2}'=010*$ ($x_{j+2}'$ is determined according to whether $x_{j+2}x_{j+3}x_{j+4}=010$), 
    and the local transition probability from $x_{j-1}'x_j'x_{j+1}'=010$ to $x_{j}x_{j+1}x_{j+2}=010$ is $1$,
    \item if $x_{j-2}x_{j-1}x_jx_{j+1}x_{j+2}=11010$, then $x_{j-2}'x_{j-1}'x_j'x_{j+1}'x_{j+2}'=1110*$, whose local transition probability is $\alpha >0$, and
    \item  if $x_{j-2}x_{j-1}x_jx_{j+1}x_{j+2}=01010$, then $x_{j-2}'x_{j-1}'x_j'x_{j+1}'x_{j+2}'=1010*$, where the local transition probability 
    from $x_{j-1}'x_j'x_{j+1}'=010$
    to $x_{j}x_{j+1}x_{j+2}=010$ is $1$, and the calculation of the local transition probability associated with $x_{j-2}x_{j-1}x_{j}=010$ reduces to the analysis of
    the local pattern `010' in $x_{j-2}x_{j-1}x_j$
\end{itemize}
(refer also to Section \ref{s:pf}, where more inductive constructions and proofs are given).

In the case $\# 010(x)= 0$ and $\# 111(x) \ge 1$, take $x'$ such that $x_j'x_{j+1}'x_{j+2}'x_{j+3}'=1011$ for an integer $j$ satisfying $x_jx_{j+1}x_{j+2}x_{j+3}=0111$, and $x_i=x_i'$ for other $i$.
Then, $f(x',x)>0$ because 
\begin{itemize}
    \item if $x_{j-1}x_jx_{j+1}x_{j+2}x_{j+3}=00111$, then $x_{j-1}'x_j'x_{j+1}'x_{j+2}'x_{j+3}'=01011$, whose local transition probability is positive,
    \item if $x_{j-1}x_jx_{j+1}x_{j+2}x_{j+3}=10111$, then $x_{j-2}=1$ (since the local pattern `010' does not appear in this case), so $x_{j-2}'x_{j-1}'x_j'x_{j+1}'x_{j+2}'x_{j+3}'=111011$, whose local transition probability is positive.
\end{itemize}
This completes the proof.
\end{proof}

Furthermore, in our setting, every recurrent irreducible component is aperiodic under the identification by a shift as follows.
Denote by $S$ the left-shift operator on $\{0,1\}^L$, i.e.,~$S(x)_j=x_{j+1\; \mathrm{mod} \; L}$ for $x=(x_j)_{j=1}^L$.

\begin{lem}
\label{lem:aperiodicity}
Every recurrent irreducible component $\Omega$ of the Markov chain defined by \eqref{eqsg} with Table~\ref{fluxsg} is periodic (i.e.~$r(\Omega) \ge 2$) if and only if $\# 010 (x) \ge 1$ and $\#11 (x)=0$ for all $x \in \Omega$.
\end{lem}

\begin{proof}
Notice that if $x\in\Omega$ satisfies $\# 010 (x) \ge 1$ and $\#11 (x)= 0$, then $f(x,S(x))=1$ and $x\neq S(x)$.
Thus, the ``if'' part holds.

Conversely, assume that $\# 010 (x) = 0$ or $\#11 (x)\ge 1$ for some $x\in\Omega$.
If $\#1(x)=0$ or $\#1(x)=L$, then obviously $f(x,x)=1$ and $\Omega$ is aperiodic, so we further assume that $\#1(x)\in [1,L-1]$.
Then, there is $\hat x\in\Omega$ such that $\# 010 (\hat x) = 0$ and $\#11 (\hat x)\ge 1$.
In fact, in the case $\# 010 (x) =0$, automatically $\#11(x)\ge 1$, so the claim holds with $\hat x=x$.
In the other case $\# 010(x)\ge 1$, due to the transition rule of Table \ref{fluxsg}, any `1' in the local pattern `010' of $x$ is shifted with probability $1$ under the stochastic process \eqref{eqsg}, but any `1' in the local pattern `11' of $x$ can remain unchanged with a positive probability
(more precisely, there is $\tilde x$ with $f(x,\tilde x)>0$ such that if $x_jx_{j+1}x_{j+2} =010$ then $\tilde x_{j+1}\tilde x_{j+2} =01$, and $x_jx_{j+1} =11$ then $\tilde x_{j}\tilde x_{j+1} =11$).
Hence, by repeating this transition in each step, we obtain the desired $\hat x$ in finite steps.
Therefore, for such a $\hat x$, it follows from the transition rule of Table \ref{fluxsg} that $f(\hat x,\hat x)>0$, so that $r(\Omega )=1$. This completes the proof.
\end{proof}

Consequently, the mean flux \eqref{Q} can be understood via the information of the stationary measures $p_j(x)$ on each irreducible component, see Corollary \ref{cor:0617}.
To investigate the dependency of the stationary distribution $p_j(x)$ on local pattern densities $\rho _{a_1\cdots a_l}$, we now consider a collection of illustrative examples.
In this section, we identify configurations that differ by some iteration of the left-shift operations, in order to simplify the representation of $\Omega$.

We first observe that an irreducible component is not determined solely by the system size $L$ and the two conserved quantities $\# 1$ and $\# 110$.
That is, even for fixed values of $L$, $\# 1$ and $\# 110$, there may exist configurations that cannot be transformed into one another through the dynamics.
For instance, in the case $L=10, \# 1=6, \# 110=2$ there are two irreducible sets $\Omega_1$ and $\Omega_2$ given by
\begin{equation}
\Omega_1=
\begin{array}{l}
\{0001101111,\ 0001110111,\ 0001111011,\ 0010110111,\\
\ 0010111011,\ 0011011101,\ 0011101101,\ 0101011011\},
\end{array}
\label{omega_1}
\end{equation}
\begin{equation}
\Omega_2=
\begin{array}{l}
\{0011001111,\ 0011010111,\ 0011100111,\ 0011101011,\ 0101101011\}.
\end{array}
\label{omega_2}
\end{equation}
The transition matrix of $\Omega_1$ is
\begin{equation*}
\begin{array}{c}
0001101111\\
0001110111\\
0001111011\\
0010110111\\
0010111011\\
0011011101\\
0011101101\\
0101011011
\end{array}
\left(
\begin{array}{cccccccc}
1-\alpha & 0 & 0 & 0 & 0 & \alpha & 0 & 0 \\
(1-\alpha)\alpha & (1-\alpha)^2 & 0 & 0 & 0 & \alpha^2 & (1-\alpha)\alpha & 0 \\
0 & \alpha & 1-\alpha & 0 & 0 & 0 & 0 & 0 \\
0 & 1-\alpha & 0 & 0 & 0 & 0 & \alpha & 0 \\
0 & \alpha & 1-\alpha & 0 & 0 & 0 & 0 & 0 \\
0 & 0 & 0 & 1-\alpha & 0 & 0 & 0 & \alpha \\
0 & 0 & 0 & \alpha & 1-\alpha & 0 & 0 & 0 \\
0 & 0 & 0 & 0 & 1 & 0 & 0 & 0 \\
\end{array}
\right),
\end{equation*}
where each component of the transition matrix ($a_{i,j}$) is a transition probability from the $i$-th configuration to the $j$-th configuration which is determined by the local rule of Table~\ref{fluxsg}.
The eigenvector of the matrix for eigenvalue 1 is
\begin{equation*}
\Big(\frac{1-\alpha}{\alpha^2},\ \frac{1}{\alpha^2},\frac{1-\alpha}{\alpha^2},\ \frac{1}{\alpha},\ \frac{1}{\alpha},\ \frac{1}{\alpha},\ \frac{1}{\alpha},\ 1\Big).
\end{equation*}
The transition matrix of $\Omega_2$ is
\begin{equation*}
\begin{array}{c}
0011001111\\
0011010111\\
0011100111\\
0011101011\\
0101101011\\
\end{array}
\left(
\begin{array}{ccccc}
1-\alpha & 0 & 0 & \alpha & 0 \\
1-\alpha & 0 & 0 & \alpha & 0 \\
0 & 2\alpha(1-\alpha) & (1-\alpha)^2 & 0 & \alpha^2 \\
0 & \alpha & 1-\alpha & 0 & 0 \\
0 & 0 & 1 & 0 & 0 \\
\end{array}
\right),
\end{equation*}
and its eigenvector of eigenvalue 1 is 
\begin{equation*}
\Big(\frac{2(1-\alpha)}{\alpha^2},\ \frac{2}{\alpha},\frac{1}{\alpha^2},\ \frac{2}{\alpha},\ 1\Big).
\end{equation*}

We now derive an example of an eigenvector for another case with $L=16, \# 1=11, \# 110=4$.
One of the corresponding sets of configurations is given by:
\begin{equation*}
\begin{array}{l}
\{0011011011011111,\ 0011011011101111,\ 0011011011110111,\\
\ 0011011011111011,\ 0011011101101111,\ 0011011101110111,\\
\ 0011011101111011,\ 0011011110110111,\ 0011011110111011,\\
\ 0011011111011011,\ 0011101101101111,\ 0011101101110111,\\
\ 0011101101111011,\ 0011101110110111,\ 0011101110111011,\\
\ 0011101111011011,\ 0011110110110111,\ 0011110110111011,\\
\ 0011110111011011,\ 0011111011011011,\ 0101101101101111,\\
\ 0101101101110111,\ 0101101101111011,\ 0101101110110111,\\
\ 0101101110111011,\ 0101101111011011,\ 0101110110110111,\\
\ 0101110110111011,\ 0101110111011011,\ 0101111011011011\}.
\end{array}
\end{equation*}
The eigenvector for eigenvalue 1 of the transition matrix is 
\[
\begin{aligned}
\Big(\frac{1-\alpha}{\alpha},\ \frac{1}{\alpha},\ \frac{1}{\alpha},\ \frac{1-\alpha}{\alpha},\ \frac{1}{\alpha},\ \frac{1}{(1-\alpha)\alpha},\ \frac{1}{\alpha},\ \frac{1}{\alpha},\ \frac{1}{\alpha},\ \frac{1-\alpha}{\alpha},\ \\
\frac{1}{\alpha},\ \frac{1}{(1-\alpha)\alpha},\ \frac{1}{\alpha},\ \frac{1}{(1-\alpha)\alpha},\ 
\frac{1}{(1-\alpha)\alpha},\ \frac{1}{\alpha},\ \frac{1}{\alpha},\ \frac{1}{\alpha},\ \frac{1}{\alpha},\ \frac{1-\alpha}{\alpha},\ \\
1,\ \frac{1}{1-\alpha},\ 1,\ \frac{1}{1-\alpha},\ \frac{1}{1-\alpha},\ 1,\ \frac{1}{1-\alpha},\ \frac{1}{1-\alpha},\ \frac{1}{1-\alpha},\ 1\Big).
\end{aligned}
\]
From these examples, one may conjecture that the unique stationary distribution $p(x)$ on each irreducible component $\Omega$ takes the form
\begin{equation}
p\left(x\right)\propto \frac{\alpha^{\# 010\left(x\right)}}{\left(1-\alpha\right)^{\# 1110\left(x\right)+\# 010\left(x\right)}}.
\label{ratio}
\end{equation}
\section{Main results} \label{s:m}
Now we are ready to state our main theorem, which confirms the validity of the  conjecture in \eqref{ratio}.
\begin{thm}\label{s:m1}
Let $\Omega\subset \{0,1\}^L$ be an irreducible component of the Markov chain defined by \eqref{eqsg} with Table~\ref{fluxsg}, and 
let $p(x)=p_\Omega(x)$ be the unique stationary measure supported on $\Omega$.
Then, for every
$x\in\Omega$, it holds that
\begin{equation}
p\left(x\right)=\frac{\frac{\alpha^{\# 010\left(x\right)}}{\left(1-\alpha\right)^{\# 1110\left(x\right)+\# 010\left(x\right)}}}{\sum_{k_1,k_2}{\frac{\alpha^{k_2}}{\left(1-\alpha\right)^{k_1+k_2}}N_{\Omega}\left(k_1,k_2\right)}},
\label{dist}
\end{equation}
where $N_{\Omega}(k_1,k_2)$ is a partition function on $\Omega$, that is, the number of configurations $x\in \Omega$ satisfying $\# 1110(x)=k_1$ and $\# 010(x)=k_2$. 
\end{thm}

Although the partition function $N_\Omega(k_1,k_2)$ itself depends on
the irreducible component $\Omega$, its dependence on $(k_1,k_2)$ is
governed by a universal combinatorial factor as follows.
Put
\[
R(k_2):=\#1-2\#110-k_2,
\]
and define
\begin{align}\label{eq:universal}
\mathcal{N}(k_1,k_2):=
\begin{cases}
    \binom{L-\#1-\#110}{k_2}
\binom{\#110}{k_1}\binom{R(k_2)-1}{k_1-1}\quad & (1\le k_1\le R(k_2)),\\
\binom{L-\#1-\#110}{k_2}\quad & (k_1=0,\; R(k_2)=0),\\
0\quad & (\text{otherwise})
\end{cases}
\end{align}
for nonnegative integers $k_1,k_2$.
Here and in what follows, we use the convention that
$\binom{a}{b}=0$ if $b<0$ or $b>a$.

\begin{prop}
\label{prop:component_independence}
Assume that
\[
0<\alpha<1,\qquad
L-\#1\geq1,\qquad
\#110\geq1,\qquad
\#1-2\#110\geq1.
\]
For every irreducible component $\Omega$, there exists a constant
$C_\Omega>0$, independent of $k_1$ and $k_2$, such that
\[
N_\Omega(k_1,k_2)=C_\Omega\mathcal{N}(k_1,k_2).
\]
\end{prop}
\begin{proof}
Choose a zero site as a root and list the zero sites cyclically.
For each zero site, look at the block of consecutive particles
immediately preceding it.  We assign the symbol $+$ to the zero site if
this block has length at least two, and the symbol $-$ otherwise.  The
cyclic word in the alphabet $\{+,-\}$ obtained in this way, considered
up to cyclic shifts, is called the skeleton orbit and is denoted by $O$.

A $+$-site accounts for one occurrence of the local pattern $110$.
Hence the number of $+$-sites is $\#110$, while the number of $-$-sites
is
$
L-\#1-\#110.
$
At a $+$-site, the two particles immediately preceding the zero are
regarded as fixed skeleton particles.  Any additional particles at a
$+$-site, and the possible single particle at a $-$-site, are called
free particles.

The update rule preserves the skeleton orbit $O$.  Moreover, under the
assumptions above and since $0<\alpha<1$, free particles can be
transported within a fixed skeleton with positive probability.  Hence
the irreducible components are precisely the sets $\Omega_O$ determined
by the skeleton orbits.

We first consider the case $k_1\geq1$.  For a fixed rooted skeleton word,
$\#010=k_2$ means that $k_2$ of the $L-\#1-\#110$ sites of type $-$
contain one free particle, while $\#1110=k_1$ means that $k_1$ of the
$\#110$ sites of type $+$ contain at least one free particle.
After choosing the $k_2$ sites of type $-$, the number of remaining free
particles is
$
R(k_2)=\#1-2\#110-k_2.
$
Thus, choosing the $k_1$ active sites of type $+$ and distributing these
$R(k_2)$ free particles among them, with each chosen site receiving at
least one particle, gives
exactly
$\mathcal{N}(k_1,k_2)$
configurations for each fixed rooted skeleton word.

The case $k_1=0$ is even simpler.  In this case no site of type $+$
contains a free particle.  Hence it is possible only when
$R(k_2)=0$, and then the only remaining choice is the choice of the
$k_2$ sites of type $-$, which gives
$
\mathcal{N}(0,k_2).
$
If $R(k_2)\neq0$, then no such configuration exists, and both sides are
zero.

Finally, if $\Omega=\Omega_O$, summing over the $|O|$ rooted skeleton
words and the $L$ possible root positions, and dividing by the number
$L-\#1$ of zeros in each configuration, yields
\begin{align}\label{eq:4.2a}
N_\Omega(k_1,k_2)
=
\frac{L|O|}{L-\#1}\mathcal{N}(k_1,k_2).
\end{align}
Thus the assertion holds with
\begin{align}\label{eq:4.2b}
C_\Omega=\frac{L|O|}{L-\#1}.
\end{align}
\end{proof}


We revisit the example in \eqref{omega_1} and \eqref{omega_2}.
According to the definition of $N_\Omega(k_1,k_2)$, one can easily calculate $N_{\Omega_1}(k_1,k_2)$ and $N_{\Omega_2}(k_1,k_2)$ as follows:
\[
\begin{array}{c|cccc}
(k_1,k_2) & (1,1) & (1,0) & (0,2) & (2,0)\\
\hline
N_{\Omega_1}(k_1,k_2) & 40 & 20 & 10 & 10\\
N_{\Omega_2}(k_1,k_2) & 20 & 10 & 5 & 5
\end{array}
\]
Since $L=10$, $\#1=6$ and $\#110=2$,
it holds that $L-\#1=4$ and $R(k_2)=2-k_2$.
Hence the nonzero values of the universal factor
$\mathcal{N}(k_1,k_2)$ are
\[
\mathcal{N}(1,1)
=4,\quad
\mathcal{N}(1,0)
=2,\quad
\mathcal{N}(0,2)
=1,\quad
\mathcal{N}(2,0)
=1.
\]
For all other pairs $(k_1,k_2)$, we have
$
\mathcal{N}(k_1,k_2)=0.
$
The components $\Omega_1,\Omega_2$ correspond to the skeleton orbits
\[
O_1=\operatorname{Orb}(\texttt{++--}),
\quad
O_2=\operatorname{Orb}(\texttt{+-+-}),
\]
respectively.
Since
$|O_1|=4$ and $|O_2|=2$,
\eqref{eq:4.2b} in Proposition~\ref{prop:component_independence} gives
\[
C_{\Omega_1}
=
\frac{L|O_1|}{L-\#1}
=
\frac{10\cdot4}{4}
=10,
\qquad
C_{\Omega_2}
=
\frac{L|O_2|}{L-\#1}
=
\frac{10\cdot2}{4}
=5.
\]
Therefore
we reproduce the above table for $N_{\Omega_1}(k_1,k_2)$ and $N_{\Omega_2}(k_1,k_2)$ via \eqref{eq:4.2a}.

\medskip

As a consequence of Theorem \ref{s:m1} and Proposition \ref{prop:component_independence}, we obtain the following estimate for $Q$, which closely aligns with the result \eqref{qd} on $Q_0$ in the deterministic case.

\begin{cor}
\label{cor:0617}
Assume that
\[
0<\alpha<1,\qquad
L-\#1\geq1,\qquad
\#110\geq1,\qquad
\#1-2\#110\geq1.
\]
Let the initial distribution of $v^0$ be an arbitrary probability
measure supported on
\[
\{x\in\{0,1\}^L:\#1(x)=\#1,\ \#110(x)=\#110\}.
\]
Then the mean flux $Q$ of the stochastic system \eqref{eqsg} is
independent of the initial distribution and is given by
\begin{equation}
\label{Qs}
Q=
\frac{1}{L}
\frac{
\displaystyle
\sum_{k_1,k_2}
(\alpha k_1+k_2)
\frac{\alpha^{k_2}}
{(1-\alpha)^{k_1+k_2}}
\mathcal{N}(k_1,k_2)
}{
\displaystyle
\sum_{k_1,k_2}
\frac{\alpha^{k_2}}
{(1-\alpha)^{k_1+k_2}}
\mathcal{N}(k_1,k_2)
},
\end{equation}
where $\mathcal{N}(k_1,k_2)$ is the universal factor given in \eqref{eq:universal}.
Consequently, the mean flux $Q_u=\rho_1-Q$ of the original stochastic
system \eqref{eqs} is uniquely determined by the pair of conserved
densities $(\rho_1,\rho_{110})$. Moreover,
\[
\lim_{\alpha\to1}Q_u
=
\max(2\rho_1-1,\,2\rho_{110}).
\]
\end{cor}

 Figure~\ref{fdtc} shows a comparison between theoretical values by (\ref{Qs}) and numerical values.
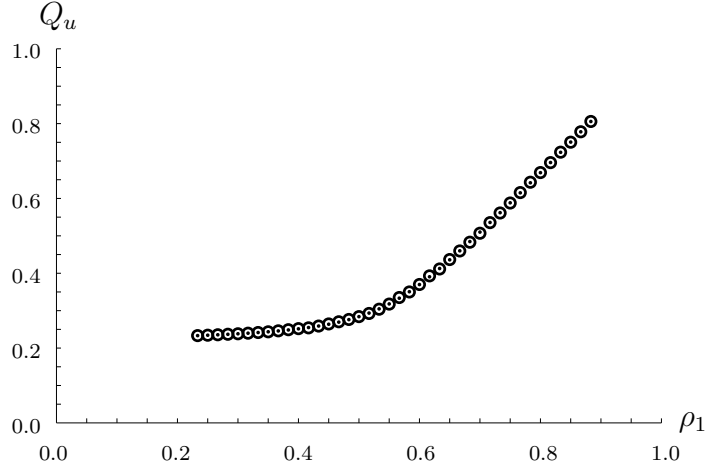
\begin{figure}[h]
\centering
  \begin{overpic}[width=90mm,trim= -15pt -15pt -15pt -15pt, clip=false]{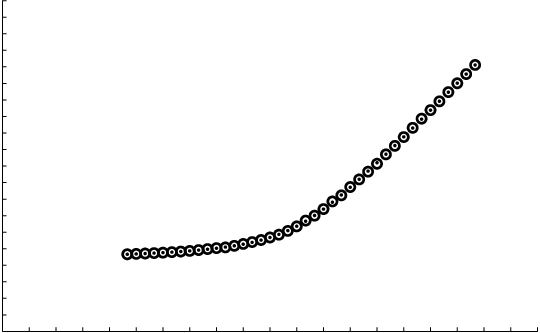} 
    \put(3,64){$Q_u$}

    \put(97,5){$\rho_1$}

    \put(3,0){\scriptsize 0.0}
    \put(21,0){\scriptsize 0.2}
    \put(39,0){\scriptsize 0.4}
    \put(57,0){\scriptsize 0.6}
    \put(75,0){\scriptsize 0.8}
    \put(93,0){\scriptsize 1.0}

    \put(-1, 4){\scriptsize 0.0}
    \put(-1,15){\scriptsize 0.2}
    \put(-1,26){\scriptsize 0.4}
    \put(-1,37){\scriptsize 0.6}
    \put(-1,48){\scriptsize 0.8}
    \put(-1,59){\scriptsize 1.0}
  \end{overpic}
\caption{Example of fundamental diagram of (\ref{eqs}) for $\rho_{110}=7/60$. Small dots ($\bullet$) are obtained by (\ref{Qs}) for $L=60, \alpha=0.7$. Circles ($\bigcirc$) are numerical results for the same parameters averaged from $n=0$ to 3000.}
\label{fdtc}
\end{figure}

\begin{proof}[Proof of Corollary \ref{cor:0617}]
In this proof, we write $Q(\nu)$ when the initial distribution of
$v^0$ is $\nu$.

First assume that $\nu$ is supported on a single irreducible component
$\Omega$ in $\{x\in\{0,1\}^L:\#1(x)=\#1,\ \#110(x)=\#110\}$.  Since $\#110\geq1$, every configuration in $\Omega$ contains the local
pattern $110$, and hence also contains the local pattern $11$.  Hence,
by Lemma~\ref{lem:aperiodicity}, the component $\Omega$ is aperiodic.

By Proposition~\ref{prop:0617}, the distribution of $v^n$ converges to
the unique stationary measure $p$ on $\Omega$.  Therefore, by
\eqref{Q} and Theorem~\ref{s:m1},
\[
Q(\nu)
=
\frac{1}{L}
\sum_{x\in\Omega}
\bigl(\alpha\#1110(x)+\#010(x)\bigr)p(x).
\]
Using the expression of $p$ in Theorem~\ref{s:m1}, this becomes
\[
Q(\nu)=
\frac{1}{L}
\frac{
\displaystyle
\sum_{k_1,k_2}
(\alpha k_1+k_2)
\frac{\alpha^{k_2}}
{(1-\alpha)^{k_1+k_2}}
N_\Omega(k_1,k_2)
}{
\displaystyle
\sum_{k_1,k_2}
\frac{\alpha^{k_2}}
{(1-\alpha)^{k_1+k_2}}
N_\Omega(k_1,k_2)
}.
\]
By Proposition~\ref{prop:component_independence},
$
N_\Omega(k_1,k_2)=C_\Omega\mathcal{N}(k_1,k_2),
$
where $C_\Omega>0$ is independent of $k_1$ and $k_2$.  This constant
cancels between the numerator and the denominator, and hence we obtain
\eqref{Qs}.  In particular, the value of $Q(\nu)$ is independent of
the irreducible component $\Omega$.

We now consider a general initial distribution $\nu$ supported on
$\{x\in\{0,1\}^L:\#1(x)=\#1,\ \#110(x)=\#110\}$.  Decompose
\[
\{x\in\{0,1\}^L:\#1(x)=\#1,\ \#110(x)=\#110\}=\bigsqcup_{\lambda\in\Lambda}\Omega_\lambda
\]
into irreducible components.  Put
$
w_\lambda:=\nu(\Omega_\lambda).
$
For $w_\lambda>0$, let $\nu_\lambda$ be the normalized restriction of
$\nu$ to $\Omega_\lambda$.  Since each $\Omega_\lambda$ is closed, the
Markov chain started from $\nu_\lambda$ stays in $\Omega_\lambda$.
By the first part of the proof,
\[
Q(\nu_\lambda)=Q_*
\]
for all $\lambda$ with $w_\lambda>0$, where $Q_*$ denotes the right-hand
side of \eqref{Qs}.  Therefore,
\[
Q
=
\sum_{\lambda\in\Lambda}w_\lambda Q(\nu_\lambda)
=
\sum_{\lambda\in\Lambda}w_\lambda Q_*
=
Q_*.
\]
Thus \eqref{Qs} holds for every initial distribution supported on
$\{x\in\{0,1\}^L:\#1(x)=\#1,\ \#110(x)=\#110\}$.

Since $Q_u=\rho_1-Q$, the mean flux $Q_u$ of the original system
\eqref{eqs} is also uniquely determined by $(\rho_1,\rho_{110})$.

It remains to prove the zero-noise limit. Notice that the maximum value of $k_1+k_2$ (i.e., the maximum value of the sum of numbers of local patterns 1110 and 010 in $\Omega$) is 
$\min{\left(L-\# 1,\# 1-2\# 110\right)}$ (cf.~\cite{3d}). 
Since the terms $\alpha^{k_2}/\left(1-\alpha\right)^{k_1+k_2}$ of maximum $k_1+k_2$ remains in the limit $\alpha \to 1$ for (\ref{Qs}), we have
\begin{eqnarray*}
\lim_{\alpha \to 1} (\rho_1-Q)&=&\frac{\# 1}{L} \\
&-&\frac{1}{L}\frac{\min{\left(L-\# 1,\# 1-2\# 110\right)}\sum_{k_1+k_2=\min{\left(L-\# 1,\# 1-2\# 110\right)}}{\mathcal{N}(k_1,k_2)}}{\sum_{k_1+k_2=\min{\left(L-\# 1,\# 1-2\# 110\right)}}{\mathcal{N}(k_1,k_2)}} \\
&=&\frac{\# 1-\min{\left(L-\# 1,\# 1-2\# 110\right)}}{L}=\rho_1-\min{\left(1-\rho_1,\rho_1-2\rho_{110}\right)}\\
&=&\rho_1+\max{\left(\rho_1-1, 2\rho_{110}-\rho_1\right)}=\max{\left(2\rho_1-1, 2\rho_{110} \right)}.
\end{eqnarray*}
This completes the proof.
\end{proof}

\section{Proof of Theorem \ref{s:m1}}\label{s:pf}
\subsection{Preliminary}
In this section, we prove Theorem \ref{s:m1}.
Multiplying both sides of \eqref{eq:1110goal} by the denominator of \eqref{dist}, it suffices to prove the equivalent identity
\begin{equation}\label{eq:1110goal2}
\sum _{x'\in \Omega} f(x',x) \frac{\alpha ^{\# 010(x')}}{(1-\alpha )^{\# 1110(x')+\# 010(x')}} =\frac{\alpha ^{\# 010(x)}}{(1-\alpha )^{\# 1110(x)+\# 010(x)}}.  
\end{equation}
The probability $f(x', x)$ depends on $\# 1110(x')$. If we denote by $\ell(x', x)$ the number of transitions from $1110$ to $1101$ during the update from $x'$ to $x$, we have
\begin{equation*}
  f(x', x)= \alpha ^{\ell(x', x)} (1-\alpha )^{\# 1110(x')-\ell(x', x)}.  
\end{equation*}
For $x\in\Omega$, put
\[
B(x):=\{x'\in\Omega:f(x',x)>0\}.
\]
Then the left-hand side of \eqref{eq:1110goal2} can be expressed as
\begin{equation}  \label{1}
\begin{aligned}
  \sum _{x'\in B(x)} \left( \frac{\alpha }{1-\alpha }\right) ^{\ell(x', x)+\# 010(x')}
\end{aligned}
\end{equation}
Thus, our goal reduces to proving the following identity:
\begin{equation}\label{eq:1110goal3}
  \sum _{x'\in B(x)} \left( \frac{\alpha }{1-\alpha }\right) ^{\ell(x', x)+\# 010(x')} = \frac{\alpha ^{\# 010(x)}}{(1-\alpha )^{\# 1110(x)+\# 010(x)}}.  
\end{equation}

\subsection{Idea of the proof of Theorem \ref{s:m1}}
To prove that the identity \eqref{eq:1110goal3} holds, we need to establish several lemmas. In this subsection, we present only an outline of the proof.
To this end, it is necessary to analyze $\ell(x', x)$, $\#010(x')$, and $B(x)$ in detail.

The quantity $\ell(x', x)$ represents the number of transitions from $1110$ to $1101$. It is important to note that not all $1101$ patterns in $x$ originate from transitions of $1110$ in $x'$. 
For example, a pattern such as $110110$ within $x$ cannot arise from any configuration that included $1110$ at the previous time step (Lemma~\ref{lemma:outside}).
On the other hand, a pattern such as $110100$ in $x$ must have originated from $11100$ and cannot result from any other pattern (Lemma~\ref{lemma:1101}).
Taking these facts into account, we classify the previous-time patterns associated with $1101$ in Table~\ref{table_of_1101}.
\begin{table}[h]
\caption{Previous-time patterns leading to $1101$}
\label{table_of_1101}
\small
\begin{center}
\begin{tabular}{lcll}
  Previous-time pattern(s) & & Current-time pattern & Transition Uniqueness\\
\hline 
    $\cdots |1110|0$ & $\to$ & $\cdots |1101|00$ & Uniquely determined \\
    & $\to$ & $\cdots |1101|10$ & Not attainable from $1110$ \\
    $\cdots |1110|1$ \quad or \quad $\cdots |1101|1$ & $\to$ & $\cdots |1101|11$ & Non-unique \\
    $\cdots |1110|10$ & $\to$ & $\cdots |1101|010$ & Uniquely determined \\
    $\cdots |1110|011$ & $\to$ & $\cdots |1101|0110$ & Uniquely determined \\
    $\cdots |1110|10$ \quad or \quad $\cdots |1110|01$ & $\to$ & $\cdots |1101|0111$ & Non-unique
\end{tabular}
\end{center}
\end{table}
We now provide a detailed explanation of the notation and interpretation used in this table.  
In Table~\ref{table_of_1101}, the symbol $|\cdot|$ is used to clarify the alignment. We omit parts irrelevant to the local pattern under consideration or those that cannot be determined from the given pattern alone. 
The column titled "Transition uniqueness" indicates how uniquely each pattern in $x$ can be traced back to a pattern in $x'$. 
"Uniquely determined" means that the current-time pattern in $x$ can arise only from the listed pattern in $x'$. 
"Non-unique" means that the pattern in $x$ may result from multiple previous-time patterns. 
"Not attainable" means that the current-time pattern cannot be generated from the corresponding pattern in $x'$.
Note that in the final row, the pattern $11010111$ can arise from two different configurations, but this does not alter the value of $\ell(x', x)$.
These characterizations will be rigorously justified in the lemmas presented in the next subsection.

Next, let us consider $\#010(x')$. The pattern $010$ deterministically transitions to either $001$ or $101$ (Lemma~\ref{lemma:v^n}(iii)). 
We classify the previous-time patterns that give rise to $001$ or $101$ in $x$, excluding cases already considered for $1101$, in Table~\ref{table_of_01}.
\begin{table}[h]
\caption{Previous-time patterns leading to $001$ or $101$}
\label{table_of_01}
\begin{center}
\begin{tabular}{lcll}
  Previous-time pattern(s) & & Current-time pattern & Transition Uniqueness\\
\hline
    $*|010|10$ \quad or \quad $*|010|0$ & $\to$ & $\cdots |001|0$ & Non-unique  \\
    & $\to$ & $\cdots |001|10$ & Not attainable from $010$ \\
    $*|010|11$ \quad or \quad $*|001|11$ & $\to$ & $\cdots |001|11$ & Non-unique  \\
    $*1|010|0$ & $\to$ & $\cdots 0|101|0$ & Uniquely determined \\
    & $\to$ & $\cdots 0|101|10$ & Not attainable from $010$ \\
    $*1|010|11$ \quad or \quad $*1|001|11$ & $\to$ & $\cdots 0|101|11$ & Non-unique  
\end{tabular}
\end{center}
\end{table}
Here, the symbol `*` denotes either $0$ or $11$. As in Table~\ref{table_of_1101}, note that the first row's $0010$ can arise from two distinct patterns in $x'$, but in both cases, the value of $\# 010(x')$ remains unchanged.

According to these classifications, for a given configuration $x$, the values of $\ell(x', x)$ and $\#010(x')$ are determined, as shown in Lemma~\ref{lemma:numberB}.
The elements of $B(x)$ are also derived based on these classifications.  
Substituting these results into \eqref{eq:1110goal3} confirms the identity. 
See \eqref{eq:proof_of_goal} for the details of the final step.

\subsection{Several Properties}
From Table~\ref{fluxsg}, the following identities for the flux function $q$ hold for all $a, b, c, d \in \{0, 1\}$.  
\begin{equation}  \label{q}
  q(a, b, c, 1)=0, \quad q(a, b, 0, d)=0, \quad q(a, 0, 1, 0)=1,  \quad q(a, 1, 0, 1)=0.
\end{equation}
Based on these, the following lemmas hold. Here, we let $x_j = v_j^{n+1}$ and $x'_j = v_j^n$ for simplicity.  
\begin{lem} \label{lemma:v^n}$\quad$
\begin{enumerate} 
  \item[(i)] If $x'_{j}x'_{j+1} = 00$, then $x_{j+1} = 0$.
  \item[(ii)] If $x'_{j} x'_{j+1} = 11$, then $x_{j} = 1$.
  \item[(iii)] If $x'_{j-1}x'_{j} x'_{j+1} = 010$, then $x_{j} x_{j+1} = 01$.
  \item[(iv)] If $x'_{j}x'_{j+1} x'_{j+2} = 101$, then $x_{j} x_{j+1} = 10$ or $01$.
  \item[(v)] If $x'_{j-2}x'_{j-1}x'_{j} x'_{j+1} = 0110$, then $x_{j-1}x_{j} x_{j+1} = 110$.  In particular, if $x'_{j-2}x'_{j-1} x'_{j} = 011$, then $x_{j-1} x_{j} = 11$.
\end{enumerate}
\end{lem} 
\begin{proof}
Each statement follows directly from \eqref{eqsg} and \eqref{q}.
\end{proof}
\begin{lem} \label{lemma:v^n+1}$\quad$
\begin{enumerate}
  \item[(i)] If $x_{j-1} x_{j} = 11$, then $x'_j = 1$.
  \item[(ii)] If $x_{j-1} x_{j} = 00$, then $x'_{j-1} = 0$.
  \item[(iii)] If $x_{j} x_{j+1} = 01$, then $x'_j x'_{j+1} = 10$ or $01$.
  \item[(iv)] If $x_{j} x_{j+1} x_{j+2} x_{j+3} = 0101$, then $x'_j x'_{j+1} = 10$.
\end{enumerate}
\end{lem}

\begin{proof}
(i)\ From \eqref{eqsg}, we have
\[
2 = x_{j-1} + x_j 
  = x'_{j-1} + x'_j + q(x'_{j-4}, x'_{j-3}, x'_{j-2}, x'_{j-1}) 
- q(x'_{j-2}, x'_{j-1}, x'_j, x'_{j+1}).
\]
Since $x'_j \in \{0,1\}$ and $q(a,b,c,d) \in \{0,1\}$, for the above equation to hold with $x'_j = 0$, it must be that
\[
x'_{j-1} = 1,\quad q(x'_{j-4}, x'_{j-3}, x'_{j-2}, 1) = 1,\quad q(x'_{j-2}, 1, 0, x'_{j+1}) = 0.
\]
However, by \eqref{q}, we have $q(a, b, c, 1) = 0$, leading to a contradiction. Therefore, $x'_j = 1$.

(ii)\ The proof is similar to (i). From \eqref{eqsg},
\[
0 = x_{j-1} + x_j 
  = x'_{j-1} + x'_j + q(x'_{j-4}, x'_{j-3}, x'_{j-2}, x'_{j-1})
 - q(x'_{j-2}, x'_{j-1}, x'_j, x'_{j+1}).
\]
For this equation to hold with $x'_{j-1} = 1$, it must be that
\[
x'_j = 0,\quad q(x'_{j-4}, x'_{j-3}, x'_{j-2}, 1) = 0,\quad q(x'_{j-2}, 1, 0, x'_{j+1}) = 1.
\]
But again, $q(a, b, 0, d) = 0$ by \eqref{q}, yielding a contradiction. Hence, $x'_{j-1} = 0$.

(iii)\ If $x'_j = x'_{j+1} = 0$, then by Lemma \ref{lemma:v^n}(i), we have $x_{j+1} = 0$, contradicting the assumption.  
If $x'_j = x'_{j+1} = 1$, then by Lemma \ref{lemma:v^n}(ii), $x_j = 1$, again contradicting the assumption.  
Therefore, the only possibilities are $x'_j x'_{j+1} = 10$ or $01$.

(iv)\ By Lemma \ref{lemma:v^n+1}(iii), the tuple $x'_j x'_{j+1} x'_{j+2} x'_{j+3}$ must be one of  
\[
1010,\quad 1001,\quad 0110,\quad 0101.
\]
In the last two cases, Lemma \ref{lemma:v^n}(v) and (iii) imply $x_{j+2} = 1$, which contradicts the assumption.  
Thus, only the case $x'_j x'_{j+1} = 10$ is valid.
\end{proof}

\subsection{Proof of (\ref{eq:1110goal3})}
To evaluate the summation in (\ref{eq:1110goal3}), we analyze $B(x)$, $\ell(x', x)$, and $\# 010(x')$ for a given configuration $x$. 
The configuration is decomposed according to the following rules:
\begin{itemize}
  \item Group every occurrence of $110$ as a unit, denoted by $(110)$.
  \item Group each $(110)1$ and $01$ as $[(110)1]$ or $[01]$.
  \item Group sequences of the form $[(110)1][01][01]\cdots[01]$ as $\{[(110)1][01][01]\cdots[01]\}$.
  \item Replace $[(110)1]$ not enclosed in $\{\}$ by $\ldbr (110)1 \rdbr$.
  \item Enclose each sequence of $[01][01]\cdots[01]$ not already within $\{\}$ by $\la [01][01]\cdots[01] \ra$.
\end{itemize}

\noindent\textbf{Example.}  For a configuration  
\begin{equation}  \label{ex of x}
0011101010101110111101100110001101001011111000000 ,
\end{equation}
it is decomposed as
{\small
\begin{equation*}
0  \la [01]\ra \{[(110)1][01][01][01]\}\ldbr (110)1\rdbr 1(110)(110)0(110)00\ldbr (110)1\rdbr 0\la [01][01]\ra 11(110)00000.  
\end{equation*}
}
The possible of previous-time local patterns that are enclosed by $\{\}$, $\ldbr \rdbr$, or $\la \ra$ in the current configuration are characterized by the following lemmas.

\begin{lem}\label{lemma:1101}
The previous-time local pattern of $\ldbr (110)1 \rdbr$ is either:
\[
1110 \quad \text{or} \quad 1101
\]
The second case occurs only when the local pattern to the right of $(110)1$ is $11$.  
\end{lem}
\begin{proof}
The predecessor patterns of $1101$ are $0101$, $0110$, $1110$, and $1101$(Lemma~\ref{lemma:v^n+1}(i) and (iii)). 
Among these, according to Lemma \ref{lemma:v^n}(iii) and (v), $0101$ and $0110$ do not result in the structure $\ldbr (110)1 \rdbr$ at the next time step.  Therefore, the first claim follows.

Suppose the symbol immediately to the right of $\ldbr (110)1 \rdbr$ is $0$. 
In that situation, the $1$ in the final $\ldbr (110)1 \rdbr$ would move one site to the right during the update(Lemma \ref{lemma:v^n}(iii)), and the resulting configuration at the next time step would no longer contain the structure $\ldbr (110)1 \rdbr$. 
Therefore, the symbol to the right must be $1$. 
Furthermore, if the next two sites are $10$, then the final $1$ in $\ldbr (110)1 \rdbr$, together with the following $0$, would form a $(110)$ pattern, contradicting the decomposition rule.
\end{proof}

\begin{lem}\label{lemma:0101}
The previous-time local pattern  of $\la [01][01] \cdots[01]\ra$ is either:
\[
  1010 \cdots1010 \quad \text{or} \quad 1010\cdots1001
\]
The second case occurs only when the pattern to the right of $\la [01][01] \cdots[01]\ra$ is $11$.  
\end{lem}

\begin{proof}
The claims follow from Lemma \ref{lemma:v^n+1}(iii) and (iv).  
The final statement can be proved in the same manner as in Lemma \ref{lemma:1101}.
\end{proof}

\begin{lem} \label{lemma:110101}
The previous-time local pattern of $\{[(110)1][01][01] \cdots[01]\}$ is either:
\[
  11101010 \cdots 1010 \quad \text{or} \quad 11101010 \cdots1001
\]
The second case occurs only when the pattern to the right of $\{[(110)1][01][01] \cdots[01]\}$ is $11$.
\end{lem}
\begin{proof}
This follows directly from Lemmas \ref{lemma:1101} and \ref{lemma:0101}.
\end{proof}

\begin{lem} \label{lemma:outside}
Any local pattern outside of $\{\}$, $\ldbr \rdbr$, and $\la \ra$ remains unchanged in the previous-time step.
\end{lem}

\begin{proof}
By the rules of decomposition, such  patterns do not contain any $(110)1$ or $01$. 
Therefore, they can be classified into  one of the following types:
\begin{enumerate}
    \item[(i)]  consecutive 1s (i.e., $111\cdots 1$),
    \item[(ii)]  a concatenation of finitely many $(110)$s and 0s,
    \item[(iii)] a concatenation of a type-(i) pattern followed immediately by a type-(ii) pattern.
\end{enumerate}
Indeed, note that one can consider such patterns as combinations of finitely many $1$, $0$, $(110)$, and that if such patterns contain $1$ not from $(110)$, then every digit left adjacent of the $1$ in the pattern should be $1$ not from $(110)$.     
In each of the following cases, we let $x_j = v_j^{n+1}$ and $x'_j = v_j^n$ for simplicity.  

\textbf{Type-(i).}  
It suffices to consider the case where $x_j = x_{j+1} = 1$, by Lemma \ref{lemma:v^n+1}(i).  
Then $x'_{j+1} = 1$.  
If we assume $x'_j = 0$, then Lemma~\ref{lemma:v^n}(i) implies that $x'_{j-1} = 1$ is necessary for $x_j = 1$.  
However, by Lemma~\ref{lemma:v^n}(iv), we then obtain $x_{j-1} x_j = 01$, which would be enclosed in $[01]$, contradicting the assumption.  
Therefore, we must have $x'_j = 1$.

\textbf{Type-(ii).}  
The case of consecutive 0s can be treated in the same way as for Type-(i).
Now consider the case of $(110)$ preceded by a $0$.  
Suppose $x_{j-1} x_j x_{j+1} x_{j+2} = 0110$.  
By Lemmas~\ref{lemma:v^n}(i), (ii), (iii), and (v), the valid predecessor patterns are $0110$ and $0111$.  
In the latter case, if $x'_{j-1} x'_j x'_{j+1} x'_{j+2} x'_{j+3} = 01111$, then $x_{j+2}$ should be 1, contradicting the assumed pattern.  
If $x'_{j-1} x'_j x'_{j+1} x'_{j+2} x'_{j+3} =01110$, then $x_{j-1} x_j x_{j+1} x_{j+2} = 0110$ with some probability $\alpha$.  
However, this would yield $x_{j+2} x_{j+3} = 01$, which would form a $[01]$ and again contradict the assumption. 

\textbf{Type-(iii).}  
It suffices to consider the case $1(110)$.  
Suppose $x_{j-1} x_j x_{j+1} x_{j+2} = 1110$.  
Then by Lemma~\ref{lemma:v^n+1}(i), we have $x'_j = x'_{j+1} = 1$.  
If $x'_{j+2} x'_{j+3} = 10$, this would produce $1(110)$ at the next time step.  
However, from the update rule, we find $x_{j+3} = 1$, so the resulting configuration becomes $1(110)1$, which would be enclosed in $\ldbr (110)1 \rdbr$, contradicting the assumption.  
From the argument for Type-(i), we also know that $x'_{j-1} = 1$.  
Hence, the only possible configuration is $1110$.

Therefore, in all cases, these patterns remain unchanged in the previous time step.  
\end{proof}

From these observations, we find that
\begin{equation}  \label{eq:numberofB}
|B(x)| = 2^{\#\}11(x) + \#\rdbr 11(x) + \#\ra 11(x)}, 
\end{equation}
where $\#\}11(x)$ denotes the number of patterns enclosed in $\{\}$ within $x$, for which the closing brace is directly followed by the pattern $11$.  
In particular, each configuration in $B(x)$ is obtained by choosing, for each such occurrence of a preceding $01$, whether to leave it as $01$ or change it to $10$.
For example, for the configuration $x$ given above:
\begin{equation*}
0011101010101110111101100110001101001011111000000 ,
\end{equation*}
one particular predecessor $x' \in B(x)$ is 
\begin{equation*}
  0\underline{01}11101010\underline{01}11\underline{01}11101100110001110010\underline{01}1111000000,  
\end{equation*}
and by replacing each underlined $01$ with $10$ respectively, we obtain $16$ distinct configurations in total, 
each of which can serve as a predecessor $x'$ that transitions to $x$.
We next characterize $\ell (x', x)$ and $\# 010(x')$ in the following lemma.  
\begin{lem} \label{lemma:numberB}
Let $\ell _1, \ell_2$ be nonnegative integers.  Then the number of configurations $x' \in B(x)$ satisfying
\[
\ell(x', x) = \# 11010(x) + \ell_1
\quad \text{and} \quad
\# 010(x') = \# 0010(x) + \# 01010(x) + \ell_2
\]
is given by 
\[
\binom{\# 110111(x)}{\ell_1}
\binom{\# 00111(x) + \# 010111(x)}{\ell_2}.
\]
\end{lem}
\begin{proof}
For $x'\in B(x)$ that satisfies the conditions stated in the lemma, define $i_1$, $i_2$, and $i_3$ to be the number of transitions where the preceding $10$ becomes $01$ just before each group $\}11(x)$, $\rdbr 11(x)$, and $\ra 11(x)$, respectively. 
Then we set $\ell_1 = i_2$, $\ell_2 = i_1 + i_3$. 
From the definitions of $\ell(x', x)$, we have
\begin{align*}
\ell(x', x) &= \#\}(x) + \#\rdbr 0(x) + i_2 = \# 11010(x) + i_2,    
\end{align*}
where $\#\}(x)$ denotes the number of patterns enclosed in $\{ \}$ within $x$.  
The first term in the middle follows from Lemma~\ref{lemma:110101}, and the second and third terms in the middle follow from Lemma~\ref{lemma:1101}.    

By Lemma~\ref{lemma:v^n}(iii), each $010$ in $x'$ results in a $01$ in $x$. These $01$ patterns in $x$ appear as 
 $[01]$, $0(110)$ or $\ldbr (110)1 \rdbr$.
For the latter two cases, Lemma~\ref{lemma:v^n} and \ref{lemma:1101} show that $010$ does not occur in the previous-time step. 
We now examine the $[01]$ pattern in detail.  Each $[01]$ appears within either $\{ \}$ or $\la \ra$.  
First, consider the case where only one $[01]$ appears, i.e., $\la [01] \ra$.  
By the decomposition rule, such a pattern appears as either $0\la [01] \ra 00$ or $0\la [01] \ra 11$.  
In the first case, Lemma~\ref{lemma:0101} implies that it gives rise to a $010$ pattern in the previous-time step.  
In the second case, a $010$ may also be generated, but this contribution is counted as part of $i_3$.  
Next, consider the $\la [01] \cdots [01] \ra$ pattern.  
Again, it suffices to consider two patterns: $0\la [01] \cdots [01] \ra 00$ and $0\la [01] \cdots [01] \ra 11$.  
In both cases, Lemma~\ref{lemma:0101} determines the structure at the previous-time step.  
The number of $010$ patterns is given by the sum of the number of $00101$ and $01010$ patterns within the pattern under consideration, plus $i_3$ for the contribution from the last $[01]$ pattern.
We finally consider the $\{(110)1[01] \cdots [01]\}$ pattern. As before, it suffices to consider two patterns: $\{(110)1[01] \cdots [01]\}00$ and $\{(110)1[01] \cdots [01]\}11$.  
In both cases, Lemma~\ref{lemma:110101} determines the structure at the previous-time step, and the number of $010$ patterns is given by the number of $01010$ pattern within  the pattern under consideration, plus $i_1$.
Note that all $01010$ patterns in $x$ must be part of either $\{(110)1[01] \cdots [01]\}0$ or $\la [01] \cdots [01] \ra 0$.  
Thus, we obtain
\begin{align*}
\#010(x') &= \#0\la [01] \ra 00(x) + \#0\la [01][01] \ra(x) + \#01010(x) + i_1 + i_3 \\
&= \#01010(x) + \#0010(x) + i_1 + i_3.
\end{align*}
Now observe that: 
\begin{align*}
\#\rdbr 11(x) &= \#\ldbr (110)1 \rdbr 11(x) = \#110111(x), \\
\#\}11(x) &= \#\{(110)1[01]\}11(x) + \#\{(110)1[01]\cdots [01]\}11(x), \\
\#\ra 11(x) &= \#0\la [01] \ra 11(x) + \#00\la [01][01] \ra 11(x) + \#10\la [01][01] \ra 11(x) \\
&\quad + \#00\la [01][01] \cdots [01] \ra(x) + \#10\la [01][01] \cdots [01] \ra(x).
\end{align*}
Note that
\[
\begin{aligned}
&\#\{(110)1[01] \cdots [01]\}11(x) + \#00\la [01][01] \cdots [01] \ra(x) + \#10\la [01][01] \cdots [01] \ra(x) \\
&= \#01010111(x).
\end{aligned}
\]
Thus, we have
\begin{align*}
\#\}11(x) + \#\ra 11(x) 
&= \#11010111(x) + \#00111(x)  \\
 &+\#00010111(x) + \#10010111(x) + \#01010111(x) \\
&=  \#00111(x)+\#010111(x).
\end{align*}

 Therefore, the number of $x' \in B(x)$ satisfying the given conditions can be counted as 
\[
\binom{\#\rdbr 11(x)}{\ell_1}
\binom{\#\}11(x) + \#\ra 11(x)}{\ell_2}
= 
\binom{\# 110111(x)}{\ell_1}
\binom{\# 00111(x) + \# 010111(x)}{\ell_2}.
\]
In particular, summing this quantity over all $\ell_1,\ell_2 \ge 0$ (with $\ell_1+\ell_2$ free) reproduces $2^{\#\rdbr 11(x)+ \#\}11(x)+\#\ra 11(x)}$, confirming that all possible predecessor configurations in $B(x)$ are accounted for. This completes the proof. 
\end{proof}
Using these lemmas, we can now simplify the left-hand side of (\ref{eq:1110goal3}) as follows:
\begin{equation} \label{eq:proof_of_goal}
\begin{aligned} 
&\sum_{x' \in B(x)} \left(\frac{\alpha}{1 - \alpha}\right)^{\ell (x', x) + \# 010(x')} \\
&= \sum_{p \geq 0} \sum_{\ell_1 + \ell_2 = p}
\binom{\# 110111(x)}{\ell_1}
\binom{\# 00111(x) + \# 010111(x)}{\ell_2} \\
& \times \left(\frac{\alpha}{1 - \alpha}\right)^{\# 11010(x) + \# 0010(x) + \# 01010(x) + p} \\
&= \left(\frac{\alpha}{1 - \alpha}\right)^{\# 010(x)}
\sum_{p \geq 0} \sum_{\ell_1 + \ell_2 = p}
\binom{\# 110111(x)}{\ell_1}
\binom{\# 00111(x) + \# 010111(x)}{\ell_2}
\left(\frac{\alpha}{1 - \alpha}\right)^p \\
&= \left(\frac{\alpha}{1 - \alpha}\right)^{\# 010(x)}
\left(\frac{1}{1 - \alpha}\right)^{\# 110111(x) + \# 00111(x) + \# 010111(x)} \\
&= \left(\frac{\alpha}{1 - \alpha}\right)^{\# 010(x)}
\left(\frac{1}{1 - \alpha}\right)^{\# 0111(x)} 
\end{aligned}
\end{equation}
Thus,  (\ref{eq:1110goal3}) holds.  \par

\section{Conclusion} \label{4}

We analyzed the asymptotic behavior of the stochastic five-neighbor cellular automaton given by (\ref{eqs}) and (\ref{eqsg}). Considering the irreducible sets of configurations, we inferred and verified an explicit formula for the stationary distribution on each irreducible component. This stationary distribution gives the asymptotic distribution on every aperiodic component. In the case of periodic component, although the distribution itself may not converge, the mean flux remains well defined along the cyclic evolution. We further showed that the mean flux is independent of the irreducible component, and hence obtained the three-dimensional fundamental diagram of the stochastic system theoretically by (\ref{Qs}). The theoretical values agree with the numerical results.  Moreover, the mean-flux formula of the deterministic system is recovered in the limit $\alpha\to1$. \par
Fundamental diagrams are generally difficult to derive for stochastic cellular automata with nontrivial conserved quantities. Our result provides a new approach to this problem: introduce two or more conserved quantities, classify the irreducible components of the configuration space by these conserved quantities, derive stationary distributions on these components, and then identify quantities, such as the mean flux, that are independent of the component. Applying this approach to other stochastic cellular automata and formulating a broader class of systems admitting exact fundamental diagrams are our future problems.
\section*{Acknowledgments}

This work was partially supported by JSPS KAKENHI Grant Numbers 22K13966, 23K03188 and 23K03056.


\bibliography{pca}
\bibliographystyle{plainurl}

\end{document}